\documentclass[11pt]{article}

\usepackage{fullpage,xcolor}

\usepackage{amsmath}
\usepackage{amssymb}
\usepackage{amsthm}

\newtheorem{theorem}{Theorem}
\newtheorem{corollary}[theorem]{Corollary}
\newtheorem{lemma}[theorem]{Lemma}
\newtheorem{conjecture}[theorem]{Conjecture}

\usepackage[colorinlistoftodos]{todonotes}

\usepackage[ruled,vlined,linesnumbered]{algorithm2e}

\usepackage{xcolor}

\usepackage{doi}
\usepackage{hyperref}

\hypersetup{ 
    pdfauthor   = {Bastien Cazaux and Eric Rivals},%
    pdftitle    = {Greedy-reduction from Shortest Linear Superstring to Shortest Circular Superstring},%
    pdfsubject  = {algorithms},%
    pdfkeywords = {superstring, NP-complete, circular string, linear string, complexity, approximation, stringology},%
    pdfcreator  = {PDFLaTeX},%
    pdfproducer = {PDFLaTeX},
    pagebackref = false,
    colorlinks = true,
    hyperfigures = true,
    urlcolor = blue,
    linkcolor=cyan,
    citecolor = green!70!black,
    pdftex
}

\title{Greedy-reduction from Shortest Linear Superstring to Shortest Circular Superstring}
\author{Bastien Cazaux and Eric Rivals\thanks{This project has received funding from the European Union's Horizon 2020 research and innovation programme under the Marie Skłodowska-Curie grant agreement No 956229.}\\
  \multicolumn{1}{p{.7\textwidth}}{\vspace{0.1cm}\centering\emph{LIRMM, University of Montpellier, CNRS, Montpellier, France}}
}
\date{December 15, 2020}

\SetKwInOut{input}{Input}
\SetKwInOut{output}{Output}

\begin{document}

\maketitle

\begin{abstract}
  A superstring of a set of strings correspond to a string which contains all the other strings as substrings. The problem of finding the Shortest Linear Superstring is a well-know and well-studied problem in stringology. We present here a variant of this problem, the Shortest Circular Superstring problem where the sought superstring is a circular string. We show a strong link between these two problems and prove that the Shortest Circular Superstring problem is NP-complete. Moreover, we propose a new conjecture on the approximation ratio of the Shortest Circular Superstring problem.
\end{abstract}

\section{Notation}

\subsection{About String}
Let $\Sigma$ be a finite alphabet, then $\Sigma^\star$ denotes the free monoid over $\Sigma$.
For a linear string $w = a_1\ldots a_n$ over the alphabet $\Sigma$, $|w| = n$ is the \emph{length} of $w$, $w[i] = a_i$ is the $i^{th}$ \emph{character} of $w$, $w[i:j] = a_i \ldots a_j$ is the \emph{substring} from the position $i$ to the position $j$. A \emph{prefix} (respectively a \emph{suffix}) is a substring which begins in $1$ (resp. which ends in $n$).
A \emph{proper substring} is a substring which differs from the string. An \emph{overlap} from a linear string $x$ to a linear string $y$ is a proper suffix of $x$ that is also a proper prefix of $y$. We denote by $ov(x,y)$ the length of the longest overlap and by $x \odot y$ the \emph{merge} from $x$ to $y$, i.e. $x \odot y = x\; y[ov(x,y)+1:|y|]$.

For a circular string $w = \langle a_1 \ldots a_n \rangle$ over the alphabet $\Sigma$, $|w| = n$ is the \emph{length} of $w$ and a \emph{substring} of $c$ is a finite substring of the linear infinite string $(a_1 \ldots a_n)^{\infty}$ (which denotes the infinite concatenation of the linear string $a_1 \ldots a_n$).

\subsection{Greedy reduction}
We will exhibit a Strict-reduction~\cite{Crescenzi97}, which is one kind of approximation-preserving reduction between two optimization problems.

A \emph{Strict-reduction} from an optimization problem $\mathcal{A}$ to another optimization problem $\mathcal{B}$ is a pair of polynomial-time computable functions $(f,g)$ where:
\begin{itemize}
\item for each instance $x$ of $\mathcal{A}$, $f(x)$ is an instance of $\mathcal{B}$,
\item for each solution $y$ of $\mathcal{B}$, $g(y)$ is a solution of $\mathcal{A}$,
\item $R_{\mathcal{A}}(x,g(y)) \leq R_{\mathcal{B}}(f(x),y)$
  \\
  where $R_{\mathcal{D}}(x,y) = \max \big(\frac{c_{\mathcal{D}}(x,OPT(x))}{c_{\mathcal{D}}(x,y)}, \frac{c_{\mathcal{D}}(x,y)}{c_{\mathcal{D}}(x,OPT(x))}  \big)$ for any optimization problem $D$ where $c_{\mathcal{D}}$ is the cost function of $\mathcal{D}$.
\end{itemize}

We propose here a new type of reduction that links the greedy nature of the solutions of two optimization problems.  Indeed, for the Strict-reduction from an optimization problem $\mathcal{A}$ to another optimization problem $\mathcal{B}$, if we have an approximation ratio of $\alpha$ for the greedy algorithm for the problem $\mathcal{B}$, we know that there exists an algorithm of $\mathcal{A}$, which can be different of the greedy algorithm for $\mathcal{A}$, with an approximation ratio smaller than or equal to $\alpha$. We want that the reduction preserves the greedy nature of the solutions.

A \emph{Greedy-reduction} from an optimization problem $\mathcal{A}$ to another optimization problem $\mathcal{B}$ is a Strict-reduction $(f,g)$ where:
\begin{itemize}
\item for each greedy solution $y$ of $\mathcal{B}$, $g(y)$ is a greedy solution of $\mathcal{A}$.
\end{itemize}

As the notion of greedy algorithm for an optimization problem may be ambiguous, we define in Appendix (Section~\ref{se:appendix}) more formally the Greedy-reduction in the case of subset system maximization problems.


\section{Superstring problems: definition and main contributions}\label{sec:super}

Let $P$ be a set of linear strings.  A \emph{linear superstring} of $P$ is a linear string which has all strings of $P$ as substring.  A circular string $\langle w \rangle$ having all strings of $P$ as substrings is a \emph{circular superstring} of $P$.
Given a set $P$ of linear strings, the \emph{Shortest Linear Superstring problem} (or \emph{SLS}) corresponds to finding the linear superstring of $P$ of minimal length.  The \emph{Shortest Circular Superstring problem} (or \emph{SCS}) is defined as finding the shortest circular superstring of $P$.
For both minimization problems, there exists a corresponding associated maximization problem, where instead of minimizing  the \emph{superstring length} measure, one maximizes the \emph{compression measure}, i.e. one seeks a superstring  maximizing the difference between $\| P \| = \sum_{w \in P} |w|$ and the length of the sought superstring.

For both problems, a specific greedy algorithm can be defined.  For the Shortest Linear Superstring problem, the well-know greedy algorithm is defined as follows: for a set $P$ of strings, the greedy algorithm takes two strings of $P$ with the maximal overlap, remove these two elements from $P$, and insert their merge into $P$ and continue until only one string remains in $P$.  This string is a greedy solution for the Shortest Linear Superstring problem for $P$.  We call this solution a \emph{greedy linear superstring}. For the Shortest Circular Superstring problem, the greedy algorithm for SCS is identical to that for SLS except that at the end, it returns the merge of the greedy linear superstring with itself, which creates a circular string. This circular string is called a \emph{greedy circular superstring}.


\begin{theorem}\label{th:L:reduction}
  There exists a Greedy-reduction from the Shortest Linear Superstring problem (SLS) to the Shortest Circular Superstring problem (SCS) for both the length and compression measures.
\end{theorem}

As SLS is NP-complete~\cite{GallantMS80}, Theorem~\ref{th:L:reduction} implies that SCS also.

\begin{corollary}
  The Shortest Circular Superstring problem is NP-complete.
\end{corollary}

Another consequence of Theorem~\ref{th:L:reduction}:
A proof of a 2-approximation of the greedy algorithm for the Shortest Circular Superstring problem would imply a proof of the well-know greedy conjecture~\cite{BlumLTY94}, which states that the approximation ratio of the greedy algorithm for the Shortest Linear Superstring problem is $2$. Hence, we propose the following conjecture:

\begin{conjecture}
  The approximation ratio of the greedy algorithm for the Shortest Circular Superstring problem is $2$.
\end{conjecture}

\section{Proof of Theorem~\ref{th:L:reduction}}

Given an ordered alphabet $\Sigma$, we take $\overline{\Sigma} = \{\overline{a} \; : \; a \in \Sigma\}$ where $\Sigma \cap \overline{\Sigma} = \emptyset$ and $\Sigma \cap \overline{\Sigma}$ is totally ordered.
Let $f$ be the function from $\mathcal{P}(\Sigma^{\star})$ to $\mathcal{P}\big((\Sigma \cup \overline{\Sigma})^{\star}\big)$ where for a set of strings $P$ over $\Sigma$, $f(P) = P \cup \overline{P}$ with $\overline{P} = \{\overline{w} : w \in P\}$ and $\overline{w} = \overline{a_1} \ldots \overline{a_k}$ for $w = a_1 \ldots a_k$.

For a circular superstring $c = \langle a_1, \ldots, a_k \rangle$ of $P \cup \overline{P}$, we denote by $l(c)$ the linear string corresponding to the smaller circular shift of $c$ in lexicographic order (for any $a\in \Sigma$ and $\overline{b} \in \overline{\Sigma}$, $a < \overline{b}$) and such that $l(c)[1] \in \Sigma$ and $l(c)[k] \in \overline{\Sigma}$. By construction, $l(c)$ exists and is unique.
For a linear string $w = a_1 \ldots a_q$ over $\Sigma'$ and $\Sigma'' \subseteq \Sigma'$, $\left.w\right|_{\Sigma''}$ is the restriction of $w$ to $\Sigma''$. We denote by $g$ the following application such that for a circular superstring $c$ of $P \cup \overline{P}$:
\[
  g(c) = a_1 \ldots a_k \text{ if } g'(c) \in \Sigma^{\ast} \text{ and } g'(c) = a_1 \ldots a_k \text{ or } \text{if } g'(c) \in \overline{\Sigma}^{\ast} \text{ and } g'(c) = \overline{a_1} \ldots \overline{a_k}
\]
with
\[
  g'(c) = \mathtt{Argmin}\Big\{|w| \; : \; w \in \{\left.l(c)\right|_{\Sigma},\left.l(c)\right|_{\overline{\Sigma}} \} \Big\}
\]

\begin{lemma}\label{le:g}
  For any circular superstring $c$ of $P \cup \overline{P}$, $g(c)$ is a linear superstring of $P$ of length smaller than or equal to $\frac{|c|}{2}$.
\end{lemma}

\begin{proof}[proof of Lemma~\ref{le:g}]
  Let $c = \langle a_1, \ldots, a_k \rangle$ be a circular superstring of $P \cup \overline{P}$.
  By definition, $g'(c) \in \{ \left.l(c)\right|_{\Sigma},\left.l(c)\right|_{\overline{\Sigma}}\}$. Assume, without loss of generality, that $g(c) = g'(c) = \left.l(c)\right|_{\Sigma}$, i.e. $|\left.l(c)\right|_{\Sigma}| \leq |\left.l(c)\right|_{\overline{\Sigma}}|$.

  As $l(c)[1] \in \Sigma$ and $l(c)[k] \in \overline{\Sigma}$, $l(c)$ is a linear superstring of $P \cup \overline{P}$ and thus $\left.l(c)\right|_{\Sigma}$ is a linear superstring of $P$. Indeed, as  $l(c)$ is a linear superstring of $P \cup \overline{P}$, for each string $s$ of $P$, there exists $i$ and $j$ such that $l(c)[i:j] = s$.
  As $l(c)[i:j] \in \Sigma^{\ast}$, there exists $i'$ and $j'$ such that $l(c)[i:j] = \left.l(c)\right|_{\Sigma}[i':j']$ and thus $s$ is a substring of $\left.l(c)\right|_{\Sigma}$. As $|l(c)| = |\left.l(c)\right|_{\Sigma}| + |\left.l(c)\right|_{\overline{\Sigma}}|$ , we have that $2|g(c)| \leq |l(c)| = |c|$.
\end{proof}

\begin{lemma}\label{le:equal}
  Let $w_o$ be a shortest linear superstring of $P$ and $c_o$ be a shortest circular superstring of $P \cup \overline{P}$.  One has $2|w_o| = |c_o|$.
\end{lemma}

\begin{proof}[Proof of Lemma~\ref{le:equal}]
  Let $w$ be a linear superstring of $P$ of length $k$. We want to prove that there exists a circular superstring of $P \cup \overline{P}$ of length smaller than or equal to $2k$. We take $c = \langle w \overline{w} \rangle$. As $w$ is a linear superstring of $P$ and $\overline{w}$ is a linear superstring of $\overline{P}$, $c$ is a circular superstring of $P \cup \overline{P}$.
  By definition, $(w \overline{w})[1:k] \in \Sigma^{\ast}$ and $(w \overline{w})[k+1:2k] \in \overline{\Sigma}^{\ast}$,  thus one gets $ov(w \overline{w},w \overline{w}) = 0$.
  Indeed, assume $ov(w \overline{w},w \overline{w}) = l > 0$.

  By construction, we know that $\Sigma \cap \overline{\Sigma} = \emptyset$.
  If $l \leq k$ then $\overline{w}[k] = w[l]$, which is impossible because $\overline{w}[k] \in \overline{\Sigma}$ and $w[l] \in \Sigma$.
  If $l>k$ then $\overline{w}[1] = w[l-k+1]$, which is also impossible since $\overline{w}[1]  \in \overline{\Sigma}$ and $w[l-k+1] \in \Sigma$.
  As $ov(w \overline{w},w \overline{w}) = 0$, $|\langle w \overline{w} \rangle| = |w \overline{w}| =|w| + |\overline{w}| = 2k$.

  Let $c$ be a circular superstring of $P \cup \overline{P}$ of length $k'$. We want to prove that there exists a linear superstring of $P$ of length smaller than or equal to $\frac{k'}{2}$.
  We take $w = g(c)$. By Lemma~\ref{le:g}, $w$ is a linear superstring of $P$ of length smaller than or equal to $\frac{|c|}{2}$, i.e. $w \leq \frac{k'}{2}$.

  Now, we can prove that there exists a shortest linear superstring of $P$ of length $k$ if and only if there exists a shortest circular superstring of $P \cup \overline{P}$ of length $2k$. Let $w_o$ a shortest linear superstring of $P$ of length $k$, we know that there exists a circular superstring of $P \cup \overline{P}$ of length smaller than or equal to $2k$ and thus there exists a shortest circular superstring of $P \cup \overline{P}$ of length smaller than of equal to $2k$. Assume the length of a shortest circular superstring of $P \cup \overline{P}$ is strictly smaller than $2k$. Then by Lemma~\ref{le:g}, there exists a linear superstring of $P$ of length strictly smaller than $\frac{2k}{2} = k = |w_o|$, which contradicts the fact that $w_o$ is a shortest linear superstring of $P$, and concludes the proof.
\end{proof}

\begin{lemma}\label{le:ineq}
  Given $w_o$ a shortest linear superstring of $P$, $c_o$ a shortest circular superstring of $P \cup \overline{P}$, and $c$ a circular superstring of $P \cup \overline{P}$, one has
  \[
    \frac{|g(c)|}{|w_o|} \leq \frac{|c|}{|c_o|} \text{ and } \frac{\|P\|-|w_o|}{\|P\|-|g(c)|} \leq \frac{\|P \cup \overline{P}\|-|c_o|}{\|P \cup \overline{P}\|-|c|}.
  \]
\end{lemma}

\begin{proof}[Proof of Lemma~\ref{le:ineq}]
  Let $w_o$ be a shortest linear superstring of $P$, $c_o$ a shortest circular superstring of $P \cup \overline{P}$ and $c$ a circular superstring of $P \cup \overline{P}$.

  By Lemma~\ref{le:g}, we have $2|g(c)| \leq |c|$ and by Lemma~\ref{le:equal}, $2|w_o| = |c_o|$ and thus $\frac{|g(c)|}{|w_o|}  = \frac{2|g(c)|}{2|w_o|} = \frac{2|g(c)|}{|c_o|} \leq \frac{|c|}{|c_o|}$.

  Moreover, because $\|P \cup \overline{P}\| = 2 \|P\|$, one gets
  \[ \frac{\|P\|-|w_o|}{\|P\|-|g(c)|} = \frac{2\|P\|-2|w_o|}{2\|P\|-2|g(c)|} \leq \frac{2\|P\|-2|w_o|}{2\|P\|-|c|} = \frac{2\|P\|-|c_o|}{2\|P\|-|c|} = \frac{\|P \cup \overline{P}\|-|c_o|}{\|P \cup \overline{P}\|-|c|}.\]
\end{proof}

Combining lemmas~\ref{le:g} and~\ref{le:ineq} gives us the Strict-reduction from SLS to SCS for both the length measure and the compression measure.

\begin{lemma}\label{le:greedy}
  Let $c$ be a greedy circular superstring of $P \cup \overline{P}$. The linear string $g(c)$ is a greedy linear superstring of $P$.
\end{lemma}

\begin{proof}[proof of Lemma~\ref{le:greedy}]
  Let $c$ be a greedy circular superstring of $P \cup \overline{P}$. As $c$ is a greedy circular superstring, there exists a greedy linear superstring $w_c$ of $P \cup \overline{P}$ such that $c$ is the merge of $w_c$ with itself. As for all strings $w \in P$ and $\overline{w} \in \overline{P}$, $ov(w,\overline{w}) = ov(\overline{w},w) = 0$, a greedy linear superstring $w$ of $P \cup \overline{P}$ has $ov(w,w) = 0$ and thus $\langle w_c \rangle = c$.
  By the definition of the greedy linear superstring, there exists $Q_1 = P \cup \overline{P}$, $Q_2$, \ldots, $Q_k = \{w_c\}$ such that $Q_i$ correspond to the $i^{th}$ recursion of the greedy algorithm for the Shortest Linear Superstring problem where $Q_{i+1} = Q_i \setminus \{u_i,v_i\} \cup \{u_i \odot v_i\}$ where $u_i$ and $v_i$ are the greedy choice at the step $i$.
  As the greedy choice takes the maximal overlap, we have $ov(u_1,v_1) \geq ov(u_2,v_2) \geq \ldots \geq ov(u_{k-1},v_{k-1})$. As $w \in P$ and $\overline{w} \in \overline{P}$, $ov(w,\overline{w}) = ov(\overline{w},w) = 0$, the set $\{i \; : \; ov(u_i,v_i) = 0\}$ is not empty and we take $j = \min(\{i \; : \; ov(u_i,v_i) = 0\})$.
  By construction, any string of $P \cup \overline{P}$ is substring of a string of $Q_j$ and each string of $Q_j$ is either in $\Sigma^{\ast}$ or in $\overline{\Sigma}^{\ast}$.
  As for all strings $w \in P$ and $\overline{w} \in \overline{P}$, $ov(w,\overline{w}) = ov(\overline{w},w) = 0$, any concatenation of strings of $Q_j \cap \Sigma^{\ast}$ is a greedy linear superstring of $P$, and similarly, any concatenation of strings of $Q_j \cap \overline{\Sigma}^{\ast}$ is a greedy linear superstring of $\overline{P}$. Hence, $g(c)$ is a greedy linear superstring of $P$.
\end{proof}

Lemma~\ref{le:greedy} shows that for any greedy circular superstring $c$ of $P \cup \overline{P}$, $g(c)$ is a greedy linear superstring of $P$. This concludes the proof of Theorem~\ref{th:L:reduction}.

\section{Appendix}\label{se:appendix}

\subsection{Greedy-reduction for subset system maximization problems}

We introduced the notion of Greedy reduction. To make it a useful concept, it is crucial to clarify what is a greedy algorithm. Especially for SLS, the greedy algorithm can be written as described above (see Section~\ref{sec:super} and \cite{GallantMS80}) or as the greedy algorithm from a specific subset system \cite{CazauxR16}
Here, we propose a definition of Greedy reduction for a subset system of maximization problems, for which the greedy algorithm is unambiguously defined.

For a finite set $E$, a subset system $\mathcal{L}$ is a set of subsets of $E$ satisfying two conditions: first, $\emptyset \in \mathcal{L}$, and second, if $B \in \mathcal{L}$ and $A \subseteq B$ then $A \in \mathcal{L}$. We denote by $\max(\mathcal{L})$ the set of elements of $\mathcal{L}$ that is maximal for inclusion.

A maximization problem is called a \emph{subset system maximization problem} if every instance of this problem defines a subset system~\cite{Mestre06}. An instance of a subset system maximization problem is thus a triplet $(E,\mathcal{L},w)$ where $E$ is a finite set, $\mathcal{L}$ is a subset system, and
 $w$ is a function that assigns a weight to each element of $E$.

An optimal solution of this problem for this instance is an element of $\mathcal{L}$ with the maximum weight, i.e. $\mathtt{Argmax}_{F\in \mathcal{L}}\big(w(F)\big)$ where $w(F) = \sum_{x \in F} w(x)$. For a given instance $(E,\mathcal{L},w)$, one can also uniquely define a greedy algorithm (see Algorithm~\ref{alg:glouton:sh}).

\begin{algorithm}[htbp]
  \input{ $(E,\mathcal{L},w)$}
  The elements $e_i$ of $E$ sorted by decreasing weight: $w(e_1) \geq w(e_2) \geq \ldots \geq w(e_n)$

  $F \leftarrow \emptyset$

  \For{$i=1$ to $n$} {
    \lIf{$F \cup \{e_i\} \in \mathcal{L}$}	{$F \leftarrow F \cup \{e_i\}$}
  }
  \Return $F$

  \output{A set $F$ of $\max(\mathcal{L})$.}

  \caption{The greedy algorithm associated with an instance $(E,\mathcal{L},w)$ of a subset system maximization problem.\label{alg:glouton:sh}}
\end{algorithm}

A Greedy-reduction from a subset system maximization problem $\mathcal{A}$ to another subset system maximization problem $\mathcal{B}$ is a pair of polynomial-time computable functions $(f,g)$ where:
\begin{itemize}
\item for each instance $(E,\mathcal{L},w)$ of $\mathcal{A}$, there exists $\mathcal{L}'$ and $w'$ such that $(f(E),\mathcal{L}',w')$ is an instance of $\mathcal{B}$,
\item for each element $y$ of $\max(\mathcal{L}')$, $g(y)$ is an element of $\max(\mathcal{L})$,
\item for each greedy solution $y$ of $(f(E),\mathcal{L}',w')$, $g(y)$ is a greedy solution of $(E,\mathcal{L},w)$,
\item $ \frac{w(g(y))}{\max_{F\in \mathcal{L}}\big( w(F)\big)} \leq  \frac{w'(y)}{\max_{F\in \mathcal{L}'}\big( w'(F)\big)}$.
\end{itemize}

Now, we can reuse the subset systems of~\cite{CazauxR16} to define the greedy algorithm for the Shortest Linear Superstring problem and for the Shortest Circular Superstring problem.

For a set of strings $P$, we denote $E_P$ the set of all pairs of $P$, i.e. $E_P = P \times P$.
We define the following subset system:
\begin{itemize}
\item $(E_P, \mathcal{L} = \{F : F \text{ satisfies (L1), (L2) and (L3)}\})$ for the Shortest Linear Superstring problem,
\item $(E_P, \mathcal{C} = \{F : F \text{ satisfies (L1), (L2) and (L3b)}\})$ for the Shortest Circular Superstring problem.
\end{itemize}
where
\begin{itemize}
\item[(L1)] $\forall s_i, \ s_j$ and $s_k \in P$, $(s_i,s_k)$ and $(s_j,s_k) \in F \Rightarrow i=j$,
\item[(L2)] $\forall s_i, \ s_j$ and $s_k \in P$, $(s_k,s_i)$ and $(s_k,s_j) \in F \Rightarrow i=j$,
\item[(L3)] for any $r \in \{1,\ldots,|P|\}$, there exists no cycle $\big((s_{i_1},s_{i_2}),\; \ldots,\; (s_{i_{r-1}},s_{i_r}),\; (s_{i_r},s_{i_1})\big)$ in $F$.
\item[(L3b)] for any $r \in \{1,\ldots,|P|-1\}$, there exists no cycle $\big((s_{i_1},s_{i_2}),\; \ldots,\; (s_{i_{r-1}},s_{i_r}),\; (s_{i_r},s_{i_1})\big)$ in $F$.
\end{itemize}

Unlike in (L3), the condition (L3b) allows for a single cycle that contains all the elements (but disallows any other cycle).

\bibliographystyle{plainurl} 
\bibliography{scs}

\end{document}